\documentclass{article}
\usepackage{graphicx} 
\usepackage{authblk}
\usepackage[T1]{fontenc}
\usepackage{soul, comment} 
\usepackage{phfqit}
\usepackage{proba} 
\usepackage[sort,nocompress]{cite}
\usepackage{amssymb}
\usepackage{amsmath}
\usepackage{amsthm}
\usepackage{amsfonts}
\usepackage{mathtools}
\usepackage{xspace}
\usepackage{bm}
\usepackage{nicefrac}
\usepackage{commath}
\usepackage{bbm}
\usepackage{boxedminipage}
\usepackage{xparse}
\usepackage[usenames,dvipsnames]{color}
\usepackage[dvipsnames]{xcolor}
\usepackage{float}
\usepackage{multirow}
\usepackage{makecell}
\usepackage{graphicx}
\usepackage{caption}
\usepackage{subcaption}
\usepackage{ifthen}
\usepackage{thm-restate}
\usepackage{algorithm}
\usepackage[noend]{algpseudocode}
\usepackage{colortbl} 
\usepackage{fancyhdr}   
\usepackage{hyperref}
    \hypersetup{ colorlinks=true, linkcolor=blue, citecolor=magenta, urlcolor=blue,}
\usepackage{fullpage}
\usepackage[utf8]{inputenc}
\usepackage{environ}
\usepackage{mdframed}
\usepackage{cleveref}
\usepackage[short]{optidef} 
\usepackage{enumitem}
\usepackage{tikz}
\usetikzlibrary{shapes}
\usetikzlibrary{positioning}
\usetikzlibrary{fit}
\usetikzlibrary{graphs,graphs.standard}
\usepackage{enumitem}

\newtheorem{proposition}{Proposition}[section]
\newtheorem{lemma}{Lemma}[section]

\newtheorem{theorem}{Theorem}[section]
\newtheorem{corollary}{Corollary}[section]

\newtheorem{definition}{Definition}[section]
\newtheorem{conjecture}{Conjecture}[section]

\usepackage{enumitem}

\NewEnviron{problem}[1]{%
	\begin{center}\fbox{\parbox{6in}{%
				{\centering\scshape #1\par}%
				\parskip=1ex
				\everypar{\hangindent=1em}%
				\BODY
}}\end{center}}

\title{A Polynomial Time Characterization For Strongly EFX Orientable Graphs}
\author{Jinghan A. Zeng \footnote{Email: \texttt{jazeng2@illinois.edu}. Siebel School of Computing and Data Science, University of Illinois Urbana-Champaign. The author is supported by an NSF Graduate Research Fellowship.}}
\date{September 2026}

\begin{document}

\maketitle
\begin{abstract}
    Discrete fair division is the problem of dividing a discrete set of goods among agents in a fair manner. In this setting, one of the most sought-after notions of fairness is envy-freeness up to any good (EFX). In 2023, Christodoulou, Fiat, Koutsoupias, and Sgouritsa introduced the idea of a graphical valuation, where the fair division problem is represented by a simple graph where vertices are the agents and the edges are the goods, and each vertex only values incident edges. They showed that an EFX allocation always exists, while determining the existence of an EFX orientation is NP-hard. They posed a open question of determining which graphs always admit an EFX orientation regardless of valuation. 

    These graphs, called strongly EFX orientable graphs, were first studied by Zeng and Mehta in 2025, who demonstrated that all such graphs have chromatic number at most 3, and bipartite graphs always admit an EFX orientation regardless of valuation. In this manuscript, we finish resolving this question by giving a polynomial time characterization of strongly EFX orientable graphs. In particular, we show that a connected graph $G$ is strongly EFX orientable if and only if either of the following is true: (1) $G$ is bipartite, or (2) the block decomposition of $G$ contains exactly one nonbipartite block $B$, and there exists a vertex $v \in B$ such that the degree of $v$ within $B$ is 2 and $G-v$ is bipartite. This proof was discovered by AI, with human intervention to break the problem into the appropriate subproblems. 
\end{abstract}
\noindent{\bf AI Disclaimer:} The proof of the characterization was generated by GPT-6 Astra, with human contribution to break the problem down into the appropriate cases. This manuscript is entirely human-written, with the proofs explained and verified by the author. The author takes full responsibility for these proofs.
\section{Introduction}
One major problem in algorithmic game theory is the question of how to allocate goods and resources to agents in a \emph{fair} manner. In the discrete fair division setting, there are $m$ goods and $n$ agents, and an allocation is a partition (with parts allowed to be empty) $[m] = X_1 \sqcup X_2...\sqcup X_n$ where $X_i$ is a \emph{bundle} that is \emph{assigned} to agent $i$. Each agent $i \in [n]$ has its own opinion on how much their bundle is worth, represented by a function $f_i:2^{[m]} \rightarrow \mathbb{R}_{\geq 0}$, mapping each subset of the goods to a numerical value. We assume that for all $i \in [n]$, $f_i$ is monotone, that is $X \subseteq Y$ implies $f_i(X) \leq f_i(Y)$ for $X,Y \subseteq [m]$. In simpler terms, a monotone valuation means that giving more items to the agent cannot make them unhappier. 

One of the most desirable definitions of "fair" is \emph{envy-freeness}, where no agent values another agent's bundle over their own. In formal terms, an allocation is said to be envy-free if for every $i,j$, we have that $f_i(X_i) \geq f_i(X_j)$. Clearly, envy-freeness is not always possible, as dividing one good among two people means that the person who did not receive the good will envy the other person that did. The second best thing we want to achieve is \emph{envy-freeness up to any good} \cite{Caragiannis2019}, also abbreviated as EFX, which means that each agent does not envy any strict subset of another agent's bundle. So this means that for all $i,j$, for all $X \subset X_j$ we have $f_i(X_i) \geq f_i(X)$. A recent counterexample \cite{akrami2026counterexampleefxnge} shows that an EFX allocation does not always have to exist for monotone valuations. 

In 2023, the setting of fair division over a graph \cite{Christodoulou23} was studied by Christodoulou, Fiat, Koutsoupias, and Sgouritsa. In this setting, the vertices are the agents, the edges are the goods, and each vertex only values edges incident to it. The authors showed that EFX allocations always exist in this setting, but determining whether an EFX orientation, which is an EFX allocation where each edge can only be assigned to its endpoint vertices, is NP-hard. The authors posed the following open question, asking which graphs will always admit an EFX orientation for any valuation on the edges. 

These graphs, later coined \emph{strongly EFX-orientable graphs}, were first studied in \cite{ZM25}, where it was shown that all bipartite graphs are strongly EFX-orientable, and all strongly EFX orientable graphs had a chromatic number of at most 3. In this manuscript, we completely resolve this question, giving a characterization of strongly EFX-orientable graphs, which shows that one can decide whether a graph is strongly EFX-orientable in polynomial time. The following characterization is below. Throughout this paper, we will assume unless otherwise stated, that all graphs are simple. We can assume, without loss of generality, that we are only concerned with characterizing connected graphs, because we can handle EFX orientations on disconnected graphs by handling each connected component separately.  

\begin{theorem}
    A connected graph $G$ is strongly EFX orientable if and only if either of the following is true:
    \begin{enumerate}
        \item $G$ is bipartite, 
        \item The block decomposition of $G$ has exactly one nonbipartite block $B$, and there exists some $v \in B$ with $d_B(v) = 2$ such that $G-v$ is bipartite, where $d_B(v)$ denotes the degree of $v$ in $B$.
    \end{enumerate}
\end{theorem}

\section{Preliminaries}

\subsection{Fair Division Definitions}

In this section, we give the relevant definitions and background information to understand both the characterization and its proof. A fair division problem is represented by $(F,m,n)$ with $n$ agents, $m$ goods, and $F=\{f_1,f_2...f_n\}$ are the valuation functions which map $2^{[m]}$ to $\mathbb{R}_{\geq 0}$. We assume all valuation functions are monotone, that is for any $i \in [n]$ and $X \subseteq Y \subseteq [m]$, we have $f_i(X) \leq f_i(Y)$. A valuation function is said to be additive if for any $S=\{x_1,x_2...x_k\} \subseteq [m]$ and $i \in [n]$, we have that $f_i(S) = \sum_{j \in [k]}f_i(\{x_j\})$, or in other words the value of a bundle is the sum of the values of the individual items on the bundle. We say that agent $i$ \emph{envies} agent $j$ if $f_i(X_i) < f_i(X_j)$. 

\begin{definition}
    An allocation $[m] = X_1 \sqcup X_2.. \sqcup X_n$ is envy-free up to any good (EFX) if for any $i,j \in [n]$ and any $g \in X_j$, $f_i(X_i) \geq f_i(X_j-\{g\})$. 
\end{definition}

We now define graphical valuations. Take a graph $G$ with vertices labeled $1,2...n$ and edges $1,2...m$, and consider the fair division setting $(F,m,n)$ on the graph. A valuation $F$ is said to be graphical on $G$, if for any $S \subseteq [m]$ and $i \in [n]$, we have $f_i(S) = f_i(S \cap E(i))$, where $E(i)$ denotes the edges incident to vertex $i$. In other words, the only edges a vertex can value are its incident edges. An additive graphical valuation on a graph $G$ is symmetric if for any $uv \in G$, we have that $f_u(uv)=f_v(uv)$. An \emph{orientation} in a graph is assigning a direction to each edge by designating one incident vertex as its "head" and the other as its "tail." We can think of an orientation as "giving" the edge to the agent representing the "head" vertex. We say that an EFX allocation in this setting is an EFX-orientation if for every edge $uv \in G$, we have that either $uv \in X_u$ or $uv \in X_v$. In other words, any edge (good) is allocated to one of its two endpoint vertices (agents). 

\begin{definition}
    A graph $G$ is strongly EFX-orientable if there exists an EFX orientation for any monotone graphical valuation over $G$. 
\end{definition}

We note that a graph being strongly EFX-orientable is a property intrinsic to the graph itself, while an EFX-orientation on a graph only makes sense in the context of both the graph and an assigned valuation. As we will see, determining whether a graph is strongly EFX-orientable or not can be done in polynomial time, in contrast to it being NP-hard to determine whether there is an EFX orientation on a graph with a valuation given. 

\subsection{Connectivity Definitions}

We will now define the relevant graph connectivity definitions needed for both the characterization and the proof. A graph is called $k$-connected, or alternatively referred to as $k$-vertex-connected, if it has at least $k+1$ vertices and is still connected after the deletion of any set of $k-1$ vertices. A complete graph on $k$ vertices is considered $(k-1)$-connected but not $k$-connected. Menger's Theorem states that a graph $G$ is $k$-connected if and only if for every distinct $s,t \in G$, we can find $k$ paths between $s$ and $t$ which overlap only at $s$ and $t$, which are referred to as \emph{internally disjoint paths}. We will use these two equivalent definitions of $k$-connectivity interchangeably. 

A \emph{cut-vertex} of a connected graph $G$ is a vertex $v$ such that $G-v$ is disconnected. Note that a graph with at least 3 vertices having no cut-vertex is equivalent to the graph being 2-connected. For graphs with cut vertices we can construct the \emph{block decomposition} of a graph. Call a graph weakly 2-connected if it is either 2-connected or a single edge. A \emph{block} of a graph is a maximal weakly 2-connected subgraph of $G$. Let $B_1,B_2...B_l$ be the blocks of $G$. These blocks form a tree-like structure which is reflected in the \emph{block cut tree} of the graph. The block cut tree is constructed by taking vertices representing $B_1,B_2...B_l$ and the cut vertices in the graph. All edges are between a cut vertex $v$ and a block $B_i$, and $v$ and $B_i$ are adjacent if and only $v \in B_i$ on the tree. This is connected as $G$ is connected, and if a cycle exists in the tree it would imply that the blocks on the cycle form a 2-connected subgraph together, contradicting the fact that a block is already a maximal 2-connected subgraph. So, the block cut tree is indeed a tree. A vertex in $G$ appears in multiple blocks if and only if it is a cut vertex, and two blocks overlap on at most one of those cut vertices. 

\begin{figure}[h]
\centering
\begin{subfigure}{.5\textwidth}
  \centering
  \includegraphics[width=.8\textwidth]{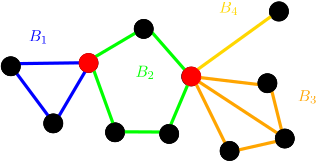}
  \label{}
\end{subfigure}%
\begin{subfigure}{.5\textwidth}
  \centering
  \includegraphics[width=.55\textwidth]{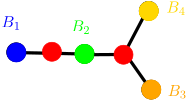}
  \label{}
\end{subfigure}
\caption{}
\label{}
\end{figure}

The following figure is an example of a block decomposition on a graph and the underlying tree structure formed by the block decomposition. Different colors are used to denote different blocks, and note that the red vertices denote the cut vertices, which appear in multiple blocks, and hence disconnect the graph when deleted. Note that a single edge (a $K_2$) can also count as a block as well. 

\subsection{Preliminary Examples}

In this section, we go over two previous results on strongly EFX-orientable graphs to have a better sense of the problem before we dive into the main proof of the characterization. 

\begin{lemma}[\cite{ZM25}]\label{lem:bipartite}
    If $G$ is bipartite, then $G$ is strongly EFX-orientable. 
\end{lemma}

\begin{proof}
    If the graph is a single vertex this is trivial, so we assume there are at least 2 vertices in $G$. Let $A,B$ be the partite sets of $G$, and take any monotone valuation on the edges. Let each vertex $a \in A$ pick its favorite edge (breaking ties arbitrarily) and orient that edge towards $a$. We note that no two $a,a' \in A$ pick the same edge, since $A$ is an independent set. For the remaining edges, orient them towards their endpoint in $B$. 

    We first claim that no vertex $a \in A$ envies another vertex. Indeed, since $G$ is simple, any vertex $v \in G$ can receive at most one edge that is valuable to $a$. However, if $v$ received the edge $av$, then this means that $a$ picked a different edge $e$, and since $e$ is the favorite edge of $a$, we have that $f_a(\{e\}) \geq f_a(\{av\})$, hence $a$ does not envy $v$. Otherwise, $v$ receives no edges valuable to $a$, and then clearly $a$ does not envy $v$ in this case as well. Two vertices $b,b' \in B$ do not envy each other, because $B$ is an independent set, so it is impossible for $b'$ to receive an edge valuable to $b$. Then, the only envy that exists is between a vertex in $B$ envying a vertex in $A$. However, since all vertices in $A$ received at most one edge, a strict subset of their bundle is the empty set, hence this orientation is EFX, as desired.  
\end{proof}

We now give an example that appears in \cite{Christodoulou23} of a graph and valuation that admits no EFX orientation. Take the $K_4$ below, where the solid edges mean the edge has a valuation of 1 to both endpoints, and dotted edges mean the edge has a valuation of 0 to both endpoints. Valuations are assumed to be additive. 

    \begin{figure}[h]
        \centering
        \includegraphics[width=0.25\textwidth]{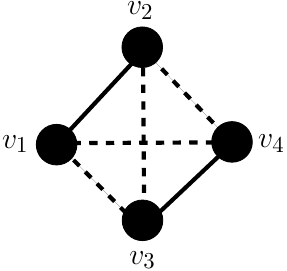}
        \caption{}
        \label{}
    \end{figure}

Note that for any orientation, no matter how we orient the edges of weight 1, we will have two vertices that are envied. Say, without loss of generality, that $v_2$ and $v_4$ receive the edges of weight 1. Note that regardless of how the edges of weight 0 are oriented, $v_1$ will envy $v_2$ and $v_3$ will envy $v_4$. Moreover, the edge $v_2v_4$ needs to go to one of $v_2$ or $v_4$, and say it goes to $v_2$. Then the bundle $v_2$ receives contains $v_2v_1$ and $v_2v_4$, but $v_1$ already envies a strict subset of the bundle $v_2$ received, namely $\{v_2v_1\}$, so this orientation is not EFX. The same logic applies if we orient $v_2v_4$ towards $v_4$ or orient the two edges of weight 1 differently, because regardless we will run into the issue of orienting an edge of weight 0 between two envied vertices, which makes obtaining an EFX orientation impossible.

We will prove that certain classes of graphs are not strongly EFX orientable by finding "bad" structures like the $K_4$. In order to do this, we will need the following proposition.

\begin{proposition}[Monotonicity of Strong EFX Orientability]
Let $G$ be a strongly EFX orientable graph, and let $H$ be a subgraph of $G$. Then $H$ is strongly EFX orientable. 
\end{proposition}

\begin{proof}
    Take a valuation on the edges of $H$. Copy that valuation to the edges of $H$ in $G$, and set the edges in $G-H$ to be worthless to both endpoints. We can find an EFX orientation on $G$, and we claim that this orientation is EFX when all edges outside of $H$ are deleted. Indeed, note that for any vertex $v$ and bundle $X_u$, $f_v(X_u)$, that is, vertex $v$'s valuation of the bundle $X_u$, remains the same on $H$, since edges outside of $H$ are worthless. Since all the valuations remain the same, the set of envied vertices in $H$ remains the same. Then, since the bundle each envied vertex receives on $H$ is a subset of its bundle on $G$, this orientation on $H$ is EFX, as desired. 
\end{proof}

\subsection{Related Work}
The idea of fair division over a graph is first introduced in \cite{Christodoulou23}, where it is shown that EFX allocations always exist over simple graphs, and determining whether a simple graph with an assigned valuation has an EFX orientation is NP-hard. \cite{ZM25} show that simple bipartite graphs always admit EFX orientations. On the other hand, it was shown in \cite{hsu2025efxorientationsmultigraphs} and \cite{Afshinmehr2025} that it is NP-hard to determine whether an EFX orientation exists on a bipartite multigraph with an assigned valuation. Furthermore, \cite{Blazej} show that it is NP-hard to determine whether an EFX orientation exists on a simple graph that is one edge-deletion away from being bipartite. \cite{ijcai2025p7} show that the EFX Orientation Problem still remains NP-hard even for graphs of vertex cover number 8, and \cite{kanellopoulos2026efxorientationsparameterizedcomplexity} improve this to show the problem still remains NP-hard for graphs of vertex cover number 4.

There has been much work on finding EFX allocations on combinatorial structures more general than simple graphs. \cite{bhaskar_et_al:LIPIcs.FSTTCS.2025.15} demonstrate that EFX allocations always exist on bipartite multigraphs with cancelable valuations, and \cite{Afshinmehr2025} show that EFX allocations on bipartite multigraphs always exist for general monotone valuations. \cite{bhaskar_et_al:LIPIcs.FSTTCS.2025.15} also prove that any $t$-colorable multigraph with girth $2t-1$ admits an EFX allocation for cancelable valuations. \cite{sgourista} show that multigraphs of girth 6 admit an EFX allocation, which is improved in \cite{afshinmehr2}, who show that all triangle free multigraphs have an EFX allocation. \cite{afshinmehr2026efxallocationsexistmultigraphs} demonstrate that EFX allocations always exist on multigraphs with cancelable valuations. \cite{lianeas2026efxallocationmultihypergraphs} demonstrate that EFX allocations exist on simple hypergraphs with girth at least 4. 

\section{Characterization of Strongly EFX Orientable Graphs}

We have the following characterization of strongly EFX orientable graphs:

\begin{theorem}
    A simple connected graph $G$ is strongly EFX orientable if and only if either of the following is true:
    \begin{enumerate}
        \item $G$ is bipartite, 
        \item The block decomposition of $G$ has exactly one nonbipartite block $B$, and there exists some $v \in B$ with $d_B(v) = 2$ such that $G-v$ is bipartite, where $d_B(v)$ denotes the degree of $v$ in $B$.
    \end{enumerate}
\end{theorem}

In this section, we will go over the proof of this theorem, and use it to show that deciding whether a graph is strongly EFX orientable or not can be done in polynomial time. The discovery of this proof was mostly done by GPT-6 Astra, and I will try to be as transparent as possible with my thought process and prompting throughout the entire session. 

In 2025, we had proved that every bipartite graph was strongly EFX orientable, moreover, strongly EFX orientable graphs have the property $\chi(G) \leq 3$. There were examples of graphs of chromatic number 3 that were strongly EFX orientable, and examples of graphs of chromatic number 3 that were not strongly EFX orientable. One observation is that the graphs of $\chi(G)=3$ that were strongly EFX orientable were very close to being bipartite graphs, but we were unable to resolve this case and hence shelved the problem. 

Recently I fed the paper into Astra and asked it to try to completely resolve this problem. It did not make much progress at all during its first attempt but handled the special case where $G$ consisted of two vertices $s,t$ and all other vertices and edges formed a collection of internally disjoint paths. This setup reminded me of connectivity, and there was a lemma in our previous work which also had to do with certain connectivity properties affecting whether a graph was EFX orientable or not. 

\begin{lemma}[\cite{ZM25}]\label{lem:suf-tight}
    Let $G$ be a bipartite graph with $|V(G)| \geq 4$ that remains connected after the deletion of any vertex. Suppose that an edge $e$ is added, which connects two vertices of the same partite set. Then $G+e$ is not strongly EFX-orientable.  
\end{lemma} 

This lemma indicated that perhaps, high connectivity was an inherent obstruction for nonbipartite graphs to be strongly EFX orientable. I then prompted GPT to first solve this problem but for graphs of high connectivity. It then mostly resolved the entire problem, showing that for any 3-connected graph, being strongly EFX orientable is equivalent to being bipartite. We will prove this using the following two results:

\begin{lemma}[\cite{ZM25}]\label{cor:fs1}
    Odd cycles that share exactly one edge, odd cycles that share exactly one vertex, and disjoint odd cycles connected by a path are not strongly EFX-orientable.
\end{lemma}  

\begin{theorem}[Tutte's Cycle Space Theorem \cite{Tutte}]
    The set of induced non-separating cycles of a 3–connected graph $G$ generates the cycle space of $G$.
\end{theorem}

We will define these terms and explain this theorem. Each subgraph $H \subseteq G$ can be encoded as a 0-1 vector with entries in $\mathbb{Z}_2$ of length $|E(G)|$, where element $i$ indicates whether edge $e_i$ is in $H$ or not. The cycle space of $G$ is the vector subspace generated by these indicator vectors for all the cycles. Tutte's Cycle Space Theorem states that one can generate the same subspace by only using the indicator vectors for the induced cycles that do not disconnect the graph when deleted. 

\begin{lemma}
    A 3-connected graph is strongly EFX orientable if and only if it is bipartite. 
\end{lemma}

\begin{proof}
    The "if" direction is obvious (see Lemma \ref{lem:bipartite}), and hence we focus on the "only if" direction. Take a 3-connected graph $G$ that is not bipartite, and take its set of induced non-separating cycles. We claim that one of the cycles there must be odd. Indeed, assume for contradiction that all such cycles are even. Then by Tutte's Cycle Space Theorem, their indicator vectors generate the entire cycle space of $G$, which we know contains an odd cycle. However, since all indicator vectors of the induced nonseparating cycles have an even number of 1s, this means that the vector space generated by them only has vectors of even parity. However, the indicator vector corresponding to the odd cycle has an odd parity yet is in the cycle space, a contradiction.  

    Then, we can find an odd cycle $C$ such that $C$ is induced and $G-V(C)$ is still connected. Assume that $G-V(C)$ is not bipartite. Then this implies the existence of two disjoint odd cycles in $G$, which means that we have two disjoint odd cycles connected by a path, a forbidden substructure by Lemma \ref{cor:fs1}. Otherwise, we have that $G-V(C)$ is bipartite. Note that in $G$, every vertex in $C$ has a neighbor in $G-V(C)$, as $G$ is 3-connected so all vertices in $G$ have a degree of at least 3. We add $C$ back to $G-V(C)$, and for each vertex $v$ in $C$, we pick a single edge $e = uv$ where $e \in G$ and $u \in G-V(C)$, and add $e$ back to the graph. We properly 2-color all of $G-V(C)$, and for a vertex in $C$, we color it the opposite color of its neighbor in $G-V(C)$. We can find $v_1,v_2 \in C$ such that $v_1,v_2$ are adjacent and have the same color. Let $u_1$ be the neighbor of $v_1$ in $G-V(C)$, and $u_2$ be the neighbor of $v_2$ in $G-V(C)$. Since $u_1,u_2$ have the same color, we can find a path $P'$ between $u_1,u_2$ on $G-V(C)$ of even length. Then $P=v_1u_1 + P' + u_2v_2$ is a path of even length from $v_1$ to $v_2$ not using any edges in $C$. Then $P+v_1v_2$ is an odd cycle, and it shares exactly one edge, $v_1v_2$, with $C$. This is a forbidden substructure by Lemma \ref{cor:fs1} as well, so this graph is not strongly EFX orientable, as desired. 
\end{proof}

While resolving the case for 3-connected graphs, Astra came up with the following conjecture. 

\begin{conjecture}
    A 2-connected nonbipartite graph is strongly EFX orientable if and only if it has a degree-2 vertex $v$ such that $G-v$ is bipartite. 
\end{conjecture}

The "if" direction of this conjecture was already known. I suggested to Astra to prove this conjecture by splitting it into cases based on the odd cycle transversal $\tau(G)$ of the graph, where $\tau(G)$ indicates the minimum number of vertices that need to be deleted to make $G$ bipartite. If a nonbipartite graph $G$ does not have a degree 2 vertex $v$ such that $G-v$ was bipartite, either $\tau(G) \geq 2$, or $\tau(G)=1$ but any vertex contained in all odd cycles had a degree of at least 3. To handle the case with $\tau \geq 2$, Astra came up with a collection of additional forbidden substructures, using the following lemma as a base. 

\begin{lemma}[\cite{ZM25}]
    A graph $G$ is EFX-orientable for every 0-1 symmetric additive valuation on the edges if and only if for every forest $F \subseteq G$ with trees $T_1,T_2...T_k$, there exists a collection of vertices, $v_1 \in T_1, v_2\in T_2...v_k\in T_k$ (we will refer to such vertices as \emph{witnesses}) such that the graph induced by $\bigcup_{i \in [k]}N_F(v_i)$ in $G$ is an independent set. 
\end{lemma}

    \begin{figure}[h]
        \centering
        \includegraphics[width=0.4\textwidth]{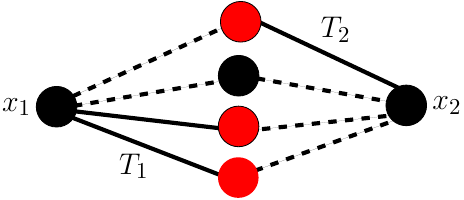}
        \caption{}
        \label{EFXfig}
    \end{figure}

As this lemma isn't straightforward to digest, we will first explain it with an example. Suppose, in the graph in Figure \ref{EFXfig}, an adversary gave us $T_1$ and $T_2$. We would then need to pick a witness vertex for every tree, and say we pick $x_1$ for $T_1$ and $x_2$ for $T_2$. The red vertices on the graph represent the vertices that are adjacent to a witness vertex on $T_1 + T_2$, and since all the red vertices form an independent set in the entire graph, this is a valid choice of witness vertices. This characterization says that for any collection of disjoint tree subgraphs the adversary gives us, we can find a witness vertex in every tree, such that if we take all the vertices that are adjacent to a witness vertex in the forest, it will be an independent set on the entire graph. 

We will now explain why this condition is necessary for a graph to be EFX orientable for all 0-1 valuations. Suppose we have a collection of trees $T_1,T_2...T_k$ on $G$ with all their edges having a symmetric value of 1 and all other edges in $G$ having a symmetric value of 0. The idea behind this characterization is that every tree will have at least one vertex which did not receive an edge of weight 1, and hence is envious of its other neighbors on the tree. If there happen to be any more edges between two of the envied vertices, then these edges cannot be oriented without breaking an EFX orientation. Hence, there needs to be a way to carefully choose these witness vertices from the forest where no edge in $G$ appears between any of their forest-neighbors, otherwise this becomes an obstruction that prevents an EFX orientation. As singleton edges are also trees, we used this lemma in our paper to construct matchings that obstructed an EFX orientation, which was sufficient to prove all graphs $G$ of $\chi(G) \geq 4$ were not strongly EFX orientable. Astra extends this idea, using forests where all trees are stars to construct an obstruction for an EFX orientation. 

\begin{definition}
    A star-forest obstruction in $G$ is a subgraph $S \subseteq G$ such that $S = T_1,T_2..T_k$ is a forest where all trees are stars, and no selection of witness vertices exists. That is, for any collection of vertices $v_1,v_2...v_k$ where $v_i \in T_i$, we have that $\bigcup_{i \in [k]} N_S(v_i)$ is not an independent set in $G$. 
\end{definition}

If $G$ has a star-forest obstruction $S$, then setting the edges in $S$ to have a value of 1 and all other edges to have a value of 0 produces a weighted graph with no EFX orientation. Hence, this can be used to create even more families of forbidden substructures, and we will do so with the help of a subdivision lemma. A subdivision of $G$ is replacing some edge $xy \in G$ with $xz$ and $zy$ with $z$ being an entirely new vertex, or in other words placing a new vertex in the middle of an edge in $G$. We say $H$ is a subdivision of $G$ if we can obtain $H$ from $G$ by repeatedly applying the subdivision operation. If we show that forbidden structures still remain forbidden structures after specific subdivisions, we can generate an entire family of forbidden graphs from a single forbidden graph. 

\begin{lemma}[Double Subdivision for Star-Forest Obstructions]\label{lem:doublesub}
    Suppose $G$ has a star-forest obstruction. Subdivide an edge $xy \in G$ twice, obtaining $G'=G-xy+xa+ab+by$ where $a$ and $b$ are new vertices. Then this new graph has a star-forest obstruction. 
\end{lemma}

\begin{proof}
\begin{figure}[h]
\centering
\begin{subfigure}{.5\textwidth}
  \centering
  \includegraphics[width=.5\textwidth]{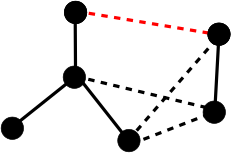}
  \label{}
\end{subfigure}%
\begin{subfigure}{.5\textwidth}
  \centering
  \includegraphics[width=.5\textwidth]{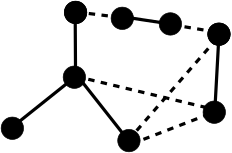}
  \label{}
\end{subfigure}
\caption{Double Subdivision of a Zero Edge}
\label{}
\end{figure}
    Let $F = T_1,T_2...T_k$ be a star forest obstruction in $G$. Suppose that $xy \not\in F$. We then claim that on $G'$, $F' =  T_1,T_2...T_k, ab$ is a star-forest obstruction. Assume not. Then we can find $v_1,v_2...v_k, u$ with $v_i \in T_i$ and $u \in \{a,b\}$ such that $N_{F'}(u) \cup \bigcup_{i \in [k]} N_{F'}(v_i) $ is an independent set in $G'$. We note that $\bigcup_{i \in [k]} N_{F'}(v_i) = \bigcup_{i \in [k]} N_{F}(v_i)$, so $\bigcup_{i \in [k]} N_{F}(v_i)$ is an independent set in $G'$ but not $G$. As $xy$ is the only edge in $G$ but not $G'$, then $xy$ must be contained in the graph induced by $\bigcup_{i \in [k]} N_{F}(v_i)$, so both $x$ and $y$ were contained in $\bigcup_{i \in [k]} N_{F'}(v_i)$ as well. However, $N_{F'}(u)$ consists of exactly one of $a$ or $b$, which means if both $x,y$ are in $\bigcup_{i \in [k]} N_{F'}(v_i)$, then $N_{F'}(u) \cup \bigcup_{i \in [k]} N_{F'}(v_i) $ is not an independent set after all, which gives us a contradiction. 

    \begin{figure}[h]
\centering
\begin{subfigure}{.5\textwidth}
  \centering
  \includegraphics[width=.5\textwidth]{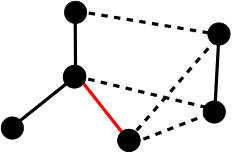}
  \label{}
\end{subfigure}%
\begin{subfigure}{.5\textwidth}
  \centering
  \includegraphics[width=.5\textwidth]{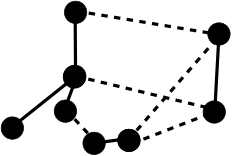}
  \label{}
\end{subfigure}
\caption{Double Subdivision of a Valuable Edge}
\label{}
\end{figure}

    Now suppose that $xy \in F$, and without loss of generality say that $xy \in T_k$ and say $x$ is the center of the star. Set $T_k'$ to be $T_k-y+xa$. We now claim that $F' = T_1,T_2..T_{k-1},T_k',by$ is a star forest obstruction. Suppose we have a collection of vertices $v_1,v_2...v_{k-1},v_k',u$ where $v_i \in T_i$ for $i \in [k-1]$, $v_k' \in T'_k$ and $u \in by$ such that $N_{F'}(u) \cup N_{F'}(v_k') \cup \bigcup_{i \in [k-1]}N_{F'}(v_i)$ is an independent set in $G'$. We first note that since $xy \in T_k$, we have that $\bigcup_{i \in [k-1]}N_{F'}(v_i) = \bigcup_{i \in [k-1]}N_{F}(v_i)$. If $v_k' = x$, then this would force $u = b$ which indicates that $\{y\} \cup  N_F'(x)  \cup \bigcup_{i \in [k-1]}N_{F'}(v_i)$ is an independent set in $G'$. We note that if we add the edge $xy$ this will still remain an independent set, as $x$ is a witness vertex and hence cannot appear in that independent set. Since we know that $N_F(x) \subseteq N_F'(x)+y$, we have that $N_F(x) \cup \bigcup_{i \in [k-1]}N_{F}(v_i)$ is an independent set on $G$. Hence, we have that $N_F'(v_k')=\{x\}$. This gives us that $\{x\} \cup \bigcup_{i \in [k-1]}N_{F'}(v_i)$ is independent in $G'$. Note that this set does not contain $y$ as $y$ is not part of $T_1,T_2..T_{k-1},T_k'$, so this set remains independent with $xy$ added to the graph. On $G$, select the witness vertices to be $v_1,v_2...v_{k-1}$ and $v_k=y$. Then we have $\bigcup_{i \in [k]}N_F(v_i) = \{x\} \cup \bigcup_{i \in [k-1]}N_{F}(v_i)$ which is an independent set in $G$, which gives us a contradiction. Hence, subdividing the same edge twice of a graph that has a star forest obstruction preserves the property that it has a star forest obstruction. 
\end{proof}

We are now ready to construct more families of forbidden substructures. 

\begin{lemma}\label{lem:stpath}
    If $s,t \in G$ have four internally vertex-disjoint paths with two of odd length and two of even length, then $G$ is not strongly EFX orientable. 
\end{lemma}

\begin{proof}

    \begin{figure}[h]
        \centering
        \includegraphics[width=0.35\textwidth]{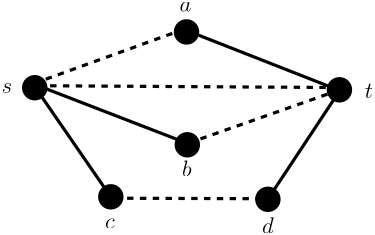}
        \caption{}
        \label{}
    \end{figure}
    Take $s,t$ and consider the graph of 4 internally disjoint paths consisting of $st$, $sat$, $sbt$, and $scdt$. (We note that one of the paths between $s$ and $t$ needs to have length at least 3, because two paths of length 1 would mean that the graph is a multi-graph and not a simple graph). Our stars to make the star-forest-obstruction will be $s,b,c$ and $t,a,d$. If the two witnesses are from $\{b,c\}$ and $\{a,d\}$ that is forbidden since $s$ and $t$ are adjacent. If $s$ is a witness, then $t$ cannot be a witness due to the edge $cd$, but neither can $a$ or $d$ since $b$ is adjacent to $t$. By a similar argument $t$ cannot be a witness, so there is no way to select witnesses for this star forest. Then, by Lemma \ref{lem:doublesub}, since adding two extra vertices to a single path will still maintain a star forest obstruction, this entire class of graphs is forbidden for strongly EFX orientable graphs. 
\end{proof}

\begin{lemma}\label{lem:K4}
    Suppose we have a $K_4$, and note that it has exactly 4 triangles. Repeatedly subdivide the $K_4$ graph. If all 4 triangles are subdivided into odd cycles, then this subdivided graph is not strongly EFX orientable. 
\end{lemma}

\begin{proof}
    Let $a,b,c,d$ be the original vertices in the $K_4$. We split this into cases based on the number of edges incident to $a$ that were subdivided into a path of odd length. 

    Case 1: All three edges incident to $a$ were subdivided into paths of odd length. Then, for the subdivided cycles of $abc$ and $acd$ to be odd cycles, the edges $bc$ and $cd$ must have also been subdivided into odd paths, which means $bd$ must have been subdivided into an odd path. Then, one can obtain this graph by taking repeated double subdivisions of a $K_4$, so we just need to prove $K_4$ has a star forest obstruction. We note that taking $ab$ and $cd$ suffices to form a star forest obstruction in a $K_4$. 

    \begin{figure}[h]
        \centering
        \includegraphics[width=0.24\textwidth]{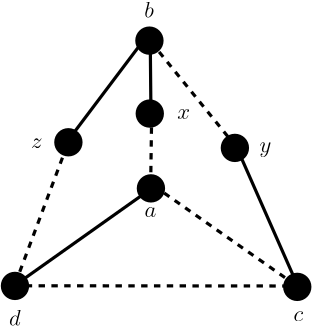}
        \caption{Star Forest Obstruction for Case 2}
        \label{}
    \end{figure}

    Case 2: Exactly one edge incident to $a$ is subdivided into a path of even length. Say that edge is $ab$, and $ac, ad$ were subdivided into odd paths. This forces $bc$ and $bd$ to be subdivided into even paths for the sake of $abc$ and $abd$ having an odd length after being subdivided. Then $cd$ is subdivided into a path of odd length as $bcd$ has odd length. Consider the $K_4$ $a,b,c,d$, but we subdivide all edges incident to $b$ exactly once, to get $bxa, byc, bzd$. Then we can obtain the graphs covered by Case 2 by double subdivision of each edge, so all that is left is to find a star forest obstruction. We pick the three stars $bx+bz$, $yc$ and $ad$. Now let us choose a set of 3 witnesses. We cannot pick $b$ as the witness, because $x$ is adjacent to $a$ and $z$ is adjacent to $d$, so whatever non witness we pick in $ad$ will share an edge with the neighborhood of $b$ in the forest. Then, either $x$ or $z$ needs to be the witness in $bx+bz$, forcing us to pick $y$ as the witness in $cy$. But $c$ is adjacent to both $a$ and $d$ meaning we have no choice in picking a witness for $ad$. Hence this is a star forest obstruction. 

    \begin{figure}[h]
        \centering
        \includegraphics[width=0.24\textwidth]{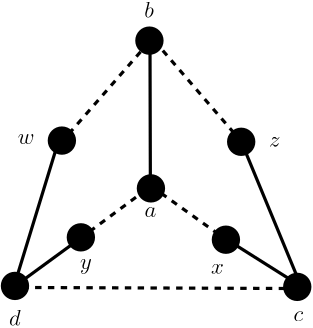}
        \caption{Star Forest Obstruction for Case 3}
        \label{}
    \end{figure}

    Case 3: Exactly two edges incident to $a$ are subdivided into a path of even length. Say $ab$ is the edge that is subdivided into an odd path, and $ac,ad$ are subdivided into even paths. Then edges $bc$ and $bd$ must be subdivided into even paths as well for $abc$ and $abd$ to be odd. Consequently, $cd$ is subdivided into a path of odd length. Take the $K_4$ $a,b,c,d$ and subdivide every edge that is not $ab$ or $cd$ once, to get $axc$, $ayd$, $bzc$, and $bwd$. Then we can recover any graph covered by this case with repeated double subdivisions, so we just need to find a star forest obstruction. We will pick the stars $ab$, $cz+cx$, and $dy+dw$. We note that picking $c$ as the witness makes it impossible to pick a witness in $ab$ by virtue of edges $bz$ and $ax$ in the graph. Hence, either $x$ or $z$ needs to be a witness, which forces $d$ to be a witness in $dy+dw$. But then this again makes it impossible to pick a witness in $ab$ due to the edges $ay$ and $bw$, hence this is a star forest obstruction. 

    Case 4: All edges incident to $a$ are subdivided into a path of even length. This forces $bc$, $bd$, and $cd$ to all be subdivided into a path of odd length. Then, note that graphs in this class are isomorphic to a graph in the class of graphs covered by Case 2, which you can see by mapping $a$ in this graph to $b$ on the respective graph covered by Case 2. We can then use the star-forest obstruction in Case 2, and we are done. 
\end{proof}

We are now ready to prove that a nonbipartite 2-connected graph $G$ with an odd cycle transversal of $\tau(G) \geq 2$ is not strongly EFX orientable. Assume, for contradiction, that this is false. We can assume that $G$ is not 3-connected, as being bipartite is necessary for a 3-connected graph to be strongly EFX orientable. Let us choose $H \subseteq G$ such that $H$ is nonbipartite, 2-connected, $\tau(H) \geq 2$, and $H$ has as few edges as possible. 

We can find vertices $s,t$ that disconnect $H$ after both are deleted. We say two (not necessarily induced) subgraphs $A,B$ form a $(s,t)$-valid partition if all of the following properties hold. 

\begin{itemize}
    \item $H = A \cup B$, and both $A$ and $B$ are connected. 
    \item $E(A) \cap E(B) = \emptyset$, and $|V(A)|,|V(B)| \geq 3$. 
    \item $V(A) \cap V(B) = \{s,t\}$.
    \item $\tau(A) \geq \tau(B)$.
\end{itemize}

\begin{lemma}\label{lem:cutset}
    For any separating set $s,t \in H$ and any $(s,t)$-valid partition $A,B$ on $H$, $A$ is not bipartite and $B$ is a path.
\end{lemma}

\begin{proof}
    If both $A$ and $B$ are bipartite, then any odd cycle in $H$ must contain an edge $x \in A$ and an edge $y \in B$ at the same time. For both edges to be part of the cycle, the cycle must pass through both $s,t$: with a traversal through one of $s,t$ needed to go outside of $A$ from $x$ and a traversal through the other needed to return. So, this means that every single odd cycle in the graph passes through both $s$ and $t$, and hence deleting one of them makes the graph bipartite, contradicting $\tau(H) \geq 2$.

    Now suppose that both $A$ and $B$ are not bipartite. Then we can find two odd cycles $C_1 \subseteq A$ and $C_2 \subseteq B$. If $C_1$ and $C_2$ meet at no more than one vertex, then either they are disjoint but connected in $H$, or they meet at exactly one vertex, both of which are forbidden by Lemma \ref{cor:fs1}. Now suppose they meet at both $s$ and $t$. We partition $C_1$ into two $s,t$ paths $P_1,P_2$, and partition $C_2$ into two $s,t$ paths $P_3,P_4$. We note that these four $s,t$ paths are internally disjoint as $A$ and $B$ as $s,t$ are the only vertices in both $A$ and $B$. Moreover, since $C_1$ is an odd cycle, the lengths of $P_1,P_2$ have different parities, and same with $P_3,P_4$. Then $P_1,P_2,P_3,P_4$ are 4 internally disjoint $s,t$-paths with 2 even and 2 odd, which is a forbidden structure by Lemma \ref{lem:stpath}. 

    So exactly one of $A$ and $B$ is not bipartite. Then $A$ must be nonbipartite. We can find a $s,t$ path $P$ of length at least $2$ in $B$. This is because if we can take a vertex $b \in B$ not in $A$, and a vertex $a \in A$ not in $B$, and there is a cycle in $H$ that passes through both by 2-connectivity. Then $s,t$ must appear on this cycle, and we can make such a $s,t$ path of length at least 2 in $B$ by taking the part of the cycle that traverses through $b$. We claim that replacing $B$ with $P$ keeps the graph both 2-connected and $\tau(A \cup P) \geq 2$. Indeed, take any $a_1,a_2 \in A$, and take some cycle $C$ that passes through both $a_1,a_2$ in $H$. If $C$ is entirely contained in $A$, then it exists on $A \cup P$ as well. Otherwise if $C$ went outside of $A$, we can split $C$ into two $s,t$ paths $P_1 \cup P_2$ with $P_1$ entirely contained in $A$. We then replace $P_2$ with $P$ to get a cycle on $A \cup P$ that passes through $a_1$ and $a_2$. Similarly, for any $b_1,b_2 \in P-\{s,t\}$, they appear on the same cycle in $A \cup P$ by taking $P$ and a path between $s,t$ contained in $A$ as the cycle. Then, to prove 2-connectivity, we need to show any $a \in A$ and $b \in P-\{s,t\}$ are also 2-connected. Indeed, note that for any vertex deleted from $A \cup P$, $b$ will still be connected to at least one of $s$ or $t$, and without loss of generality say it is $s$. Since $V(A)$ is 2-connected within $A \cup P$, $a$ can reach $b$ through $s$,  hence we have 2-connectivity of $A \cup P$.

    So now we have to prove $\tau(A \cup P) \geq 2$. Assume not, then since $A$ is nonbipartite and $B$ is bipartite, all odd cycles in $A \cup P$ pass through some vertex $x \in A$. Then, there is an odd cycle $C$ on $H$ that doesn't pass through $x$ (by virtue of $\tau(H) \geq 2)$ and hence must traverse through an edge in $B-E(P)$, but also must traverse through an edge in $A$ since $B$ is bipartite. Break $C$ into $s,t$ paths $P_1, P_2$ where $P_2$ is entirely contained in $B$ and contains an edge outside of $A$. We first note that $P_2$ must have a different parity than $P$, otherwise $P_1 \cup P$ would form an odd cycle in $A \cup P$ that doesn't traverse through $x$, contradicting the fact that all odd cycles in $A \cup P$ pass through $x$. We note that both $P,P_2$ are $s,t$ paths contained entirely in $B$. However, since $B$ is bipartite, and $s,t \in B$, this means that all paths between $s$ and $t$ in $B$ must have the same parity. So, $P,P_2$ cannot have different parities either, which is a contradiction. Hence $\tau(A \cup P) \geq 2$. Then, by minimality of $H$, we have that $B$ must be a path. 
\end{proof}

We now need to utilize Lemma \ref{lem:K4}. We now define a "splitting off" operation for vertices of degree 2. In a graph $G$, if a vertex $x$ has exactly two neighbors $u,v$, we say that a splitting-off at $x$ is forming a new \emph{multi}graph $G'=G-x+uv$, essentially creating a shortcut between $u,v$ without $x$. We can think of splitting off as the inverse operation of a subdivision, as $G$ is a subdivision of $G'$. We note that this operation, however, can turn a simple graph into a multigraph, with the possibility of self-loops as well. 

Assign all edges on $H$ a weight of 1. For a vertex $x$ with degree 2, when splitting it off by removing $xu,xv$ and adding $uv$, assign $w(uv)=w(xu)+w(xv)$. We repeatedly split off all vertices of degree 2 in $H$ in this manner to obtain a new graph $K$ with no vertices with degree 2. Note that the weight of each edge in $K$ represents the length of the original path in $H$. 

\begin{lemma}
    $K$ is a simple, 3-connected graph. 
\end{lemma}

\begin{proof}
    $K$ has no self-loops, because a self loop at a vertex $v$ would be a cycle on the original graph, which would be disconnected from $H$ once $v$ is deleted. Furthermore, $K$ has more than 2 vertices, because if $K$ consisted of 2 vertices, $H$ would be a set of internally disjoint paths between two vertices, and deleting a vertex gets rid of all of the cycles, contradicting $\tau(H) \geq 2$.

    Now, suppose there exist $a,b$ joined by two parallel edges $p,q$. Let $D$ be the subgraph on $H$ corresponding to $p,q$ and $C$ contain the rest of the edges on $H$. Note that $C$ and $D$ only overlap at $a,b$. Moreover, $a,b$ is a cutset of $H$, because there must have been a vertex originally on one of $p,q$ in $H$, and deleting both $a$ and $b$ isolates that vertex. This also implies $|V(D)| \geq 3$, and $|V(C)| \geq 3$ as $K$ has at least 3 vertices. Clearly, $D$ is connected, and $C$ must be connected too, otherwise one of $a$ or $b$ would be a cut vertex of $H$. Then one of $C,D$ or $D,C$ forms an $(a,b)$-valid partition on $H$, so by Lemma \ref{lem:cutset}, one of $C,D$ must be nonbipartite on $H$ and the other is a path. If $D$ is nonbipartite, then $C$ must be a path, but this again gives us that $H$ is a set of internally disjoint paths between $a,b$, contradicting the fact that $\tau(H) \geq 2$. But $D$ is not a path, hence a contradiction, so no multi-edges exist.  

    Then since the minimum degree of $K$ is at least 3, $K$ has at least 4 vertices. $K$ is 2-connected by virtue of $H$ being 2-connected, but suppose that $K$ has a cut-set $a,b$ that splits $K$ into $X,Y$. Let $X'$ be the graph induced by $X,a,b$ and let $Y'$ be the graph induced by $Y,a,b$ but the edge $ab$ is removed if it exists. Lift both $X',Y'$ to $H$, and note that they now form an $(a,b)$-valid partition (after correctly ordering the parts). Then one of $X',Y'$ must be lifted to a path on $H$, which meant it was contracted into a single edge on $K$. But then one of $X,Y$ would be empty, a contradiction.  
\end{proof}

We will now find a $K_4$ subdivision in $K$ where the sum of edge weights in each cycle corresponding to a triangle is odd, which will give us our desired contradiction.

\begin{lemma}
    $H$ is not strongly EFX orientable. 
\end{lemma}

\begin{proof}
    For each edge $e$ in $K$, assign it a label of 1 if $w(e)$ is odd and 0 otherwise. Note that there exists some cycle $C$ with a label sum of 1 in $\mathbb{Z}_2$, corresponding to an odd cycle in $H$. Create a function $f:\mathbb{Z}_2^{|E(K)|}\rightarrow \mathbb{Z}_2$, where the edge indicator vectors for each subgraph of $K$ are mapped to the parity of odd-labeled edges in that subgraph. We note that $f$ is a linear map. Since $K$ is 3-connected, by Tutte's Cycle space theorem, we can find a collection of induced non-separating cycles $C_1,C_2..C_l$ such that $C=C_1+C_2...+C_l$. This also means that $1=f(C)=f(C_1)+f(C_2)...+f(C_l)$ hence some cycle $C' \in \{C_1,C_2...C_l\}$ has $f(C')=1$, so its edge weights sum to an odd number. 

    Let $R$ be the subgraph of $K$ induced by $V(K)-V(C')$. Clearly $R$ has no cycle that has the sum of its edge weights odd, or this would imply that $H$ has two disjoint odd cycles. We will now change the 0-1 edge labels on $K$, but do so in a way that preserves the parity of all the cycles. We define the operation $Switch(v)$ on the graph $K$ as swapping the labels of all the edges incident to $v$. We note that this preserves the parity of all the cycles, because for any cycle that passes through $v$, exactly two edges in the cycle, $e,f$ are incident to $v$. The sum of the two edge labels was $w(e)+w(f) \mod 2$ before the switch, and is $w(e)+1+w(f)+1 \mod 2$ after the switch, which yield an equal parity. As all other edges in the cycle retained the same label, the parity of this cycle is preserved, and we can safely apply $Switch(v)$ at any vertex. 

    We will first apply the $Switch$ operation to change the label of every edge within $R$ to $0$. Construct a BFS tree through $R$ at any arbitrary root vertex $r \in R$, and give the vertices a label of $0$ or $1$ depending on the parity of the weighted distance from the vertex to $r$ on the tree. We again note that $R$ does not have any cycles of weighted odd parity. We then apply $Switch(v)$ at any vertex labeled 1 (note this may also affect some edges between $R$ and $C'$). If an edge in $R$ had both endpoints labeled 1, then it must have had a weight of 0, so both of its endpoints applying switch keeps it a weight of 0. Similarly, if an edge in $R$ had both endpoints labeled 0, then it must have had a weight of 0 and still has a weight of 0 after the switch. If an edge in $R$ has both endpoints with different labels, then such an edge must have a weight of 1, and exactly one endpoint had switch applied to it, so it must now have a weight of 0. 

    Now, for any vertex in $v \in C'$, we can ensure it is incident to an edge with label 0 that crosses between $C'$ and $R$ by applying $Switch(v)$ if it is only incident to cross edges with label 1. Note that every vertex in $C'$ must be incident to at least one cross edge as $K$ has minimum degree 3. We apply $Switch$ on the vertices in $C'$ so that all vertices are incident to a cross edge with label 0, and the number of edges with label 1 in $C'$ is maximized. Since $C'$ has odd parity, the number of edges in $C'$ with label 1 must be odd. Assume that we can find at least 3 edges with label 1 in $C'$. Select a pair $e,f$ of edges of label 1 such that the distance between the two edges is minimized. We can find vertices $a,b,c$ such that $e$ is the only edge of label 1 between $a$ and $b$, and $f$ is the only edge of label 1 between $b$ and $c$. Then $a,b,c$ partition $C'$ into three paths of odd parity. Let $aa'$, $bb'$, and $cc'$ be the cross edges of label 0. Extend these edges so they meet at $R$ in some minimal tree $T$. Since this tree has 3 leaves, it must have a vertex $r \in R$ somewhere of degree 3. We note that $C'$ and $T$ form a $K_4$ subdivision, with $r,a,b,c$ as the vertices. $C'$ is an odd cycle and all paths in $T$ from $r$ to $a,b,c$ have even parity. Then, this is a $K_4$ subdivision where the triangles are lifted to odd cycles, and since all switch operations preserve the parity of all cycles, lifting this to $H$ gives us our contradiction. 

    Then, we handle the case where $C'$ only has one odd edge, which we call $pq$. We call a vertex at $C'$ flexible if both 0,1 appear as labels on its cross edges. We note that by construction, non-flexible vertices only have 0 as a label on its cross edges. If a vertex $s \in C'$ that is not $p$ or $q$ is flexible, we can apply $Switch(s)$ so $C'$ has at least 3 odd edges, a contradiction. Both $p$ and $q$ cannot be flexible at the same time either, because applying switch on both will still keep a label of 1 at $pq$ but turn the other two incident edges to have a label of 1. Hence, all edges of label 1 on the graph must be incident to one of $p,q$ say $p$. But then this means all odd cycles in $K$ pass through $p$, and hence deleting $p$ makes $H$ bipartite, a contradiction. Hence $H$ cannot be strongly EFX-orientable. 
\end{proof}

Hence this gives us the following theorem.

\begin{theorem}
    If $G$ is strongly EFX-orientable and 2-connected, then $\tau(G) \leq 1$. 
\end{theorem}

This also gives us an alternative way of showing $\chi(G) \leq 3$ for any strongly EFX orientable graph $G$ as well. Any 2-connected strongly EFX orientable graph is 3-colorable as $\tau(G) \leq 1$. If $G$ is not 2-connected, then we take its block decomposition, and 3-color each 2-connected block starting from the root, and make our way down to the leaf blocks. We now handle the case where $G$ is 2-connected and has $\tau(G) \leq 1$. Our goal is to prove the following theorem. 

\begin{theorem}
    If $G$ is nonbipartite and 2-connected, $G$ is strongly EFX-orientable if and only if there exists some $v \in G$ with $d(v)=2$ such that $G-v$ is bipartite. 
\end{theorem}

To prove necessity, we just need to eliminate the case where $\tau(G) =1$ but all transversal vertices have degree at least 3. Sufficiency follows from the following lemma. 

\begin{lemma}[\cite{ZM25}]\label{lem:suf2}
    If $G$ is a graph with some $v \in G$ such that for any edge $e$ incident to $v$, $G-e$ is bipartite, then $G$ is strongly EFX-orientable. 
\end{lemma}

Then on a 2-connected graph with $\tau(G)=1$ with a transversal vertex $v$ with degree 2, since all odd cycles must pass through $v$, they must pass through both edges of $v$, so removing any edge of $v$ changes $G$ into a bipartite graph. Hence this gives us sufficiency, so we prove necessity. 

Take a 2-connected graph $H$ with $\tau(H)=1$, and let $T(H)$ be the set of all vertices such that for all $u \in T(H)$, $H-u$ is bipartite. Assume that all $v \in T(H)$ have degree at least 3. We wish to show an obstruction that prevents an EFX orientation. Astra was prompted to handle this with suggestions by splitting it into cases where $|T(H)| = 1$, $|T(H)| > 1$, and to possibly utilize the block decomposition of $H-v$. Take any vertex $v \in T(H)$, and delete it from $H$, noting that we now have a bipartition $H-v=(A,B)$. Add in vertices $s,t$, where $s$ has its neighborhood as $N(v) \cap B$ and $t$ has its neighborhood as $N(v) \cap A$. Note that this is essentially "cutting" $v$ into two vertices. Refer to this new graph as $H'$. We first eliminate the case where $|T(H)|=1$. 

\begin{lemma}
    If $|T(H)| = 1$, $H'$ is 2-connected. 
\end{lemma}

\begin{proof}
    We claim that $H'$ is $2-$connected. Indeed, note that if either $s$ or $t$ is deleted, then $H'-\{s,t\}$ will still be connected by virtue of $H-v$ being connected, and the other vertex in $\{s,t\}$ is connected to the rest of the graph by its edges. So, assume some vertex $u \in H'-\{s,t\}$ is deleted. For any $x,y \in H'-u-\{s,t\}$, we can take a path from $x$ to $y$ in $H-u$. If $v$ does not appear on the path, then $x,y$ are connected in $H'-u$. Otherwise, we note that $s,t$ are still connected in $H'-u$. Indeed, since $T(H)=1$, there must exist some odd cycle that does not go through $u$, then this cycle becomes an $s,t$ path on $H'$ which does not traverse through $u$. Then, for any instance of $v$ in an $x,y$ path, we can replace $v$ with $s$, $t$, or an $s,t$-path, and one of these three options constructs a valid $x,y$-path on $H'-u$. Since $s$ and $t$ are connected in $H'-u$ and at least 3 vertices in $H'$ are adjacent to one of them, $s$ and $t$ still remain connected to the rest of the graph after the deletion of $u$, hence we have 2-connectivity of $H'$. 
\end{proof}

\begin{lemma}
    If $H'$ is 2-connected for some choice of $v$, then $H$ is not strongly EFX orientable.
\end{lemma}

\begin{proof}
    We can find 2 internally disjoint paths between $s,t$ on $H'$, and since $H'$ is bipartite and $s,t$ belong to different partite sets, then these 2 paths must have an odd length. Then on $H$ they correspond to two odd cycles that intersect exactly at $v$, but this is a forbidden substructure, and hence $H$ cannot be strongly EFX orientable, as desired. 
\end{proof}

This fully resolves the case where $|T(H)|=1$. We can now tackle the case where $|T(H)| > 1$ and we can assume $H'$ is not 2-connected for any choice of $v$. 

\begin{lemma}
    The blocks in the block decomposition of $H'$ form a path in the block cut tree with $s,t$ on opposite ends. Moreover, all cut vertices of $H'$ appear in $T(H)$. 
\end{lemma}

\begin{proof}
    We first note that neither $s,t$ are cut vertices in $H'$ by virtue of $H-v$ being connected. Take the block cut tree of $H'$, and note that there is exactly one path from the block containing $s$ to the block containing $t$. Suppose for contradiction, there is another vertex on the tree not on the path. Then, there is a vertex $b$ on the tree that corresponds to a block in $H'$ not on the path between $s$ and $t$. This means that the deletion of some vertex $u$ in the block corresponding to $b$ isolates the other vertices from both $s$ and $t$ in $H'$, which means deleting $u$ isolates vertices from $v$ in $H$, contradicting the fact that $H$ is 2-connected. So, the block decomposition of $H'$ is a path with the blocks containing $s,t$ on opposite ends. If a vertex $u$ separates $s,t$ but is not in $T(H)$, then we can find an odd cycle in $H$ through $v$ that does not include $u$. Then this translates to a $s,t$ path in $H'$ that does not traverse through $u$, so $u$ cannot be a cut vertex. 
\end{proof}

\begin{figure}[h]
\centering
\begin{subfigure}[c]{.5\textwidth}
  \centering
  \includegraphics[width=.85\textwidth]{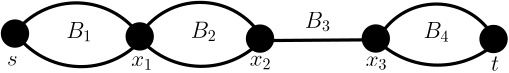}
  \label{}
\end{subfigure}%
\begin{subfigure}[c]{.5\textwidth}
  \centering
  \includegraphics[width=.55\textwidth]{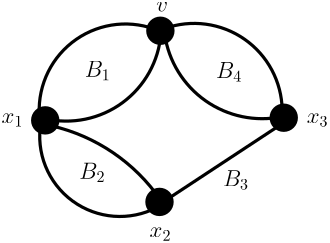}
  \label{}
\end{subfigure}
\caption{An example of what the structure of $H'$ and $H$ could look like for $k=4$.}
\label{}
\end{figure}

Let $k$ be the number of blocks in the block decomposition of $H'$. Let $x_1,x_2...x_{k-1}$ be the cut vertices along the path from $s$ to $t$ in the block cut tree, and set $x_0=s$ and $x_k=t$. We can get 2 collections of paths $P_1,P_2,...P_k$ and $Q_1,Q_2...Q_k$ with this construction: for $x_i,x_{i+1}$ if the block $x_i,x_{i+1}$ is a single edge then $P_{i+1}=Q_{i+1} = x_ix_{i+1}$. Otherwise, if the block containing $x_i,x_{i+1}$ is not a single edge, take any cycle through $x_i,x_{i+1}$, and split it into two internally disjoint paths $P_{i+1},Q_{i+1}$. We note that for each $i \in [k]$, the parity of $P_i$ and $Q_i$ are equal as $H'$ is bipartite. We glue $s,t$ back together into $v$. This structure looks like a "ring" of cycles and edges glued together.
We will now prove that this structure is not strongly EFX orientable, to where we utilize the following lemma. 

\begin{lemma}[Double Subdivision of Zero Edges \cite{ZM25}]
    Suppose there exists an additive symmetric valuation on $G$ which $G$ has no EFX orientation for. Pick an edge  $xy$ of value $0$ in the valuation, and replace it with any path of odd length to obtain $G'$. Then $G'$ is not strongly EFX orientable. 
\end{lemma}

This lemma allows us to "contract" cycles into smaller cycles to make things easier to prove. For any $i \in [k]$, if $P_i,Q_i$ are both even, replace them with two paths of length 2 with weights $(0,2)$ and $(2,0)$. If $P_i,Q_i$ are both odd but not in a single-edge-block, replace the shorter path (breaking ties arbitrarily) with a path of length 1 with weight $0$, and the other with a path of length 3 with weights $(2,0,2)$. Finally, for any blocks that are single edges, assign a weight of $1$. We note that all weight-1 edges form a matching, because if not, this would denote a vertex in $T(H)$ with degree 2 in $H$. Denote this weighted graph by $H''$. 

\begin{lemma}
    $H''$ does not have an EFX orientation. 
\end{lemma}

\begin{proof}
    Suppose for contradiction we can find an EFX orientation on $H''$. Call a vertex in $H''$ \emph{happy} if it receives an edge of weight 2 and \emph{unhappy} otherwise. 

    Let us examine a cycle made by even paths, $x-a-y$ and $x-b-y$, where $xb$ and $ay$ have a weight of $2$. Suppose $x$ is unhappy, then by definition $xb$ points to $b$. This forces $x$ to envy $b$, so $by$ points to $y$. Then if $y$ is happy, then one of its neighbors in $H''$ will envy it, but also $y$ received an edge of 0 value, preventing EFX. So, if $x$ is unhappy so is $y$. By a symmetric argument $y$ being unhappy also forces $x$ to be unhappy, so they are both happy at the same time or unhappy at the same time. We note that when both $x,y$ are unhappy, they envy $b,a$ respectively, and hence must also receive the 0 value edges $xa$ and $yb$. 

    Examine a cycle made by odd paths, $x-y$ and $x-a-b-y$. If $x$ is unhappy, then $xa$ points to $a$ by definition, which means $x$ envies $a$. Then $ab$ points towards $b$. If $y$ is unhappy then $y$ would envy $b$, which violates EFX. So $y$ must be happy, forcing $xy$ to point to $x$. If $x$ is happy, then $xy$ points to $y$ as some vertex envies $x$. Then $y$ cannot be happy because some other vertex would envy $y$ if it was. So $x$ and $y$ have opposite states, with the unhappy vertex also being forced to receive an edge of 0 value. 
    
    Finally, let us examine a single edge $x-y$. Note the other block containing $x$ and the other block containing $y$ are both cycles, as no single-edge-blocks share vertices. Suppose $x$ is happy, then some other vertex in $H''$ must envy $x$. This forces $xy$ to point to $y$, preventing $y$ from being happy. Now suppose that $x$ is unhappy. Assume that $y$ is unhappy, and orient $xy$ towards $x$. Then $y$ would envy $x$, because $y$ can only be incident to one edge of weight 1 and can't receive an edge of weight 2. But $x$ is unhappy and part of a cycle, and it was previously established that such unhappy $x$ would be forced to receive an edge of weight 0 in the cycle. Hence this violates EFX, so $y$ is happy. 

    Now, take the vertices $v,x_1,x_2...x_{k-1}$, assign a label of 0 if it is unhappy and assign a label of 1 if it is happy. Between $x_{i-1},x_i$, if $P_i$ is odd then they swap labels and if $P_i$ is even they keep the same label. Then, an even number of swaps must happen between $v,x_1,x_2...x_{k-1},v$, for $v$ to keep a consistent label. But since in $H'$, all paths between $s,t$ are odd, $P_1+P_2...+P_k$ has an odd length hence this forces an odd number of swaps, a contradiction. Hence, this graph cannot be EFX orientable with these values, and applying the valid path substitutions on the weight 0 edges gives us that $H$ is not strongly EFX orientable as well. 
\end{proof}

We have now established that for 2-connected graphs, the necessary and sufficient condition for $G$ to be strongly EFX orientable is that either $G$ is bipartite or we can find a degree-2 vertex $v$ such that $G-v$ is bipartite. The question then remains how to handle graphs with a cut vertex. We can take the block decomposition of the graph $G$. If there are at least 2 nonbipartite blocks, this indicates either two disjoint odd cycles or two odd cycles meeting at exactly one vertex, which is forbidden. So, we can assume that there is exactly one nonbipartite block $B$ if $G$ is nonbipartite yet strongly EFX orientable. 

$B$ itself is 2-connected and strongly EFX orientable, and since it is nonbipartite, it is necessary for some $v \in B$ to exist such that $B-v$ is bipartite, with $v$ having 2 neighbors in $B$. The main issue, however, is that the cut vertices in $B$ (which may potentially include $v$) may have neighbors in other blocks which may affect the existence of an EFX orientation given a valuation, so it is not immediately clear whether this condition is sufficient. Astra was able to show, however, that this condition is indeed sufficient, which we will prove here and complete our characterization. 

\begin{lemma}
    Let $G$ be a graph with a block decomposition such that exactly one block $B$ is not bipartite and there exists a vertex $v \in B$ such that $d_{B}(v)=2$ and $B-v$ is bipartite. Then $G$ is strongly EFX orientable. 
\end{lemma}

\begin{proof}
    This is trivial if $G$ is 2-connected, so we assume $G$ is not. Take any valuation on the edges of $G$. Let $xv,yv$ be the edges incident to $v$ in $B$, and assume without loss of generality that $f_v(xv) \geq f_v(yv)$. We first suppose that $xv$ is the favorite edge of $v$. We first delete $xv$ from $G$ and note that $G-xv$ is 2-colorable. Create a bipartition $A,F$ such that $x \in A$, and note that $v \in A$ as well (otherwise if $v \in F$, $G$ would be bipartite even with $xv$ added back). Assign $xv$ to $v$, and let every other vertex in $A$ pick its remaining favorite edge. Since $xv$ is the only edge with both endpoints in $A$, the rest of the unpicked edges have an endpoint in $F$, so we can orient them towards the vertex in $F$. No vertex $a \in A$ envies a vertex in $F$, as each vertex in $F$ received at most one edge $a$ cares about, which will be of at most as much value than the one it picked. Similarly, no envy exists between vertices of $F$ as $F$ is an independent set. Then, the only vertices that can be envied are vertices in $A$, but they received only one edge, so this orientation is EFX. 

    Now, assume that $xv$ is not the favorite edge of $v$. Then $v$ is a cut vertex of $G$, and let $G_1,G_2$ be a partition of the connected components of $G-v$ cut with $u \in G_2$ if and only if $u$ can reach $x,y$ after the deletion of $v$. We will orient $G_1+v$ and $G_2+v$ separately. We delete $yv$ from $G_2+v$ and bipartition the remaining graph into $A,F$ such that $x \in A$ and $y,v \in F$. Allow each vertex in $A$ to pick its favorite edge. 

    If $x$ picks a different edge than $xv$, then we assign $xv$ to $v$. We add $yv$ back and orient the rest of edges to a vertex in $F-v$. This is EFX, as no vertices in $A$ envy another vertex due to allowing them to pick their favorite edge. $F-v$ only has out-neighbors in $A$ which received exactly one edge, so no violations of EFX are induced there. $v$ does not envy $x$ or $y$ since it received its favorite edge in $G_2+v$, and is not envied since $x$ picked its favorite edge which was not $xv$. On $G_1+v$, bipartition it into $C,D$ so that $v \in D$, allow the vertices in $C$ to pick their favorite edge and orient the rest towards $D$. This gives us an EFX orientation, and since $v$ is not envied in any of the two graphs, it is safe for us to glue the two orientations at $v$, giving us our desired result. 

    Now, suppose that $x$ picks $xv$. Orient the rest of the edges (including $yv$) towards $F-v$. Note that the orientation between $A,F-v$ is clearly EFX, the main issue is that $v$ may envy $y$ yet $y$ could have received more than one edge. Bipartition $G_1+v$ into $C,D$ so that $v \in C$ and let every vertex in $C$ pick its favorite edge, and orient the rest towards $D$. Note that $v$ receives its favorite edge, because $xv$ was not its favorite edge so it must have been in $G_1+v$. Then this ensures a valid EFX orientation between $A,F$, as $v$ does not envy $y$ anymore. Moreover, since $v$ only has one incoming edge in all of $G$, this orientation on $G$ is EFX as only $A \cup C$ can be envied but all vertices there received exactly one edge. Hence, we can find an EFX orientation on $G$, which completes the proof of our characterization. 
\end{proof}

We note that we can find the block decomposition of a graph in $O(|V|+|E|)$ time \cite{tarjan}, hence this characterization yields the following corollary. 

\begin{corollary}
    Given a connected graph $G$, we can decide whether $G$ is strongly EFX orientable in $O(n(n+m))$ time, where $n$ is the number of vertices and $m$ is the number of edges in $G$. 
\end{corollary}

\begin{proof}
    Apply the algorithm to obtain the block decomposition of the graph, which takes $O(n+m)$ time. For each block, test if it is bipartite using a BFS tree, which takes time linear in the number of vertices and edges in the block. Hence, testing all of the blocks takes $O(n+m)$ time in total. If at least two blocks are nonbipartite, return that $G$ is not strongly EFX orientable. If all blocks are bipartite, return that $G$ is strongly EFX orientable. If exactly one block $B$ is not bipartite, for every $v \in B$ where $v$ has two neighbors in $B$, test whether $B-v$ is bipartite. On a success return that $G$ is strongly EFX orientable, and if all tests fail return that $G$ is not strongly EFX orientable. This step takes $O(n(n+m))$ time, so in total this algorithm runs in $O(n(n+m))$ time. 
\end{proof}

\bibliography{references}

\end{document}